\documentclass{llncs}

\usepackage[dvipdfm]{graphicx}
\usepackage{url}
\usepackage{ascmac,bm}
\usepackage{amsmath,amssymb}
\usepackage{float}
\usepackage{hyperref}
\usepackage{algorithm,algorithmic}
\usepackage{extarrows}

\makeatletter
\def\endfor{\end{ALC@for}}
\def\endif{\end{ALC@if}}
\def\endwhile{\end{ALC@while}}
\makeatother

\newcommand{\ed}{$\hfill\diamond$}
\newcommand{\ued}{\begin{flushright}\raisebox{6ex}[1ex][0ex]{{\smash
				  \ed}}\end{flushright}\vspace{-6ex}}

\makeatletter
\newcommand{\figcaption}[1]{\def\@captype{figure}\caption{#1}}
\newcommand{\tblcaption}[1]{\def\@captype{table}\caption{#1}}
\makeatother

\def\CA{{\cal A}}
\def\CB{{\cal B}}
\def\CC{{\cal C}}

\def\CD{{\cal D}}
\def\CE{{\cal E}}
\def\CF{{\cal F}}

\def\diag{\mathrm{diag}}
\def\dim{\mathrm{dim}}
\def\rank{\mathrm{rank}}
\def\nullity{\mathrm{null}}
\def\inv{\mathrm{in}}
\def\out{\mathrm{out}}

\makeatletter
  \def\@thmcountersep{.}
\makeatother

\renewcommand{\sec}[1]{Section \ref{sec:#1}}

\makeatletter
  
  \@addtoreset{equation}{section}
\makeatother

\makeatletter
  
  \@addtoreset{proposition}{section}
\makeatother

\begin{document}
\title{Graph Spectral Properties of \\Deterministic Finite Automata}
\author{Ryoma Sin'ya}
\institute{Department of Mathematical and Computing Sciences,
Tokyo Institute of Technology.\\\email{shinya.r.aa@m.titech.ac.jp}}

\maketitle
\begin{abstract}
We prove that a minimal automaton has a minimal adjacency matrix rank
and a minimal adjacency matrix nullity using equitable partition (from
graph spectra theory) and Nerode partition (from automata theory). This
result naturally introduces the notion of matrix rank into a regular
language $L$, the minimal adjacency matrix rank of a deterministic
automaton that recognises $L$.

We then define and focus on rank-one languages: the class of languages
for which the rank of minimal automaton is one. We also define the expanded
canonical automaton of a rank-one language.

\end{abstract}

\section{Introduction}\label{sec:introduction}
The {\it counting function}\footnote{Also called as {\it growth
function}, {\it generating function} or {\it combinatorial function}} $C_L:
\mathbb{N} \rightarrow \mathbb{N}$ of a language $L$ over a finite
alphabet maps a natural number $n$ into the number of words in $L$ of length $n$
defined as:
\[
 C_L(n) := |\{ w \in L \mid |w| = n \}|.
\]
The counting function is a fundamental object in formal language theory
and has been studied extensively ({\it cf.}
\cite{Shur:2008:CCR:1813695.1813728,journals/corr/abs-1010-5456})．
If $L$ is a regular language, we can represent its counting function $C_L(n)$ using
the $n$-th power of an adjacency matrix of a deterministic automaton
that recognises $L$.
Our interest is in the ``{\it easily countable}'' class of languages, in
the intuitive sense of the word.
In this paper, we define and focus on \emph{rank-one languages}: the class of
languages that can be recognised by a deterministic automaton for which
the adjacency matrix rank is one.\\

\noindent{\bf Counting and its applications}
For any regular language
$L$, it is a well-known result that the counting function of $L$
satisfies:
\begin{eqnarray}
  C_L(n) = I M^n F \label{eq:counting}
\end{eqnarray}
 where $M$ is an adjacency matrix, $I$ is an initial vector and $F$ is a
 final vector of any deterministic automaton recognises $L$, since
 $(M^n)_{ij}$ equals to the number of paths of length $n$ from $i$ to $j$
 and this corresponds the number of words  of length $n$ ({\it cf.}
 Lemma 1 in \cite{RyomaSIN'YA2013}).
 We give the simple example of Equation \eqref{eq:counting} as follows.

\begin{example}\label{ex:counting}
Let $\CA_{fib}$ is a deterministic automaton recognises $L = (a+ba)^*$
 and
\[
 M = \begin{bmatrix}
	  1 & 1\\
	  1 & 0
	 \end{bmatrix}, \;\;\;\;\;\;\;\;\;\;\;\;
 I = \begin{bmatrix}
	  1 & 0
	 \end{bmatrix}, \;\;\;\;\;\;\;\;\;\;\;\;
 F = \begin{bmatrix}
	  1\\
	  0
	 \end{bmatrix}
\]
 are its adjacency matrix, initial and final vector. Then the following holds.

\begin{minipage}{.6\columnwidth}
\begin{eqnarray}
 C_L(0) &=& |\{ \varepsilon \}| = 1, \;\;\; C_L(1) = |\{ a \}| = 1,\nonumber\\
 C_L(2) &=& |\{ aa, ba \}| = 2,\nonumber\\
 &\vdots&\nonumber\\
 C_L(n) &=& I M^n F = \begin{bmatrix}
					 1 & 0
					\end{bmatrix}
					\begin{bmatrix}
					 1 & 1\\
					 1 & 0
					\end{bmatrix}^n
					\begin{bmatrix}
					 1\\
					 0
					\end{bmatrix}\label{eq:fib}
\end{eqnarray}
\end{minipage}
\begin{minipage}{.3\columnwidth}
 \centering\includegraphics[width=1\columnwidth]{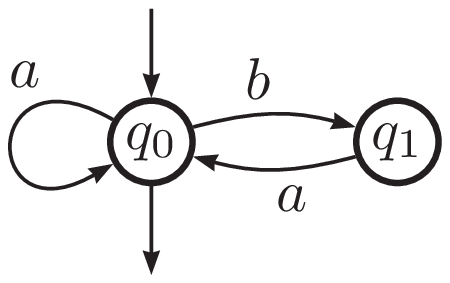}\\
 \large{\centering{$\CA_{fib}$}}
\end{minipage}\vspace{1em}\\
\noindent Equation \eqref{eq:fib} means that $C_L(n)$ equals to the
 $(n+1)$-th Fibonacci number.
\ed 
\end{example}

{\it Ranking} is one of the variants of counting.
The ranking function of $L$ over a finite alphabet $A$ is a bijective
function $R_L: L \rightarrow \mathbb{N}$ that maps a word $w$ in $L$ to its
index in the lexicographic ordering $\prec$ over $A^*$ defined as:
\[
 R_L(w) := |\{ v \in L \mid v \prec w \}|.
\]
In 1985, Goldberg and Sipser introduced a ranking-based string
compression in \cite{Goldberg:1985:CR:22145.22194}.
Recently, the author studied a ranking-based compression on a regular
language to analyse its compression ratio and improve a
ranking algorithm in \cite{RyomaSIN'YA2013}.
We show an example of a ranking-based compression on a regular language.

\begin{example}\label{ex:url}
The formal grammar of \emph{Uniform Resource Identifier} (URI) is
 defined in RFC 3986 \cite{bernerslee2005uri}, and it is known
 that the formal grammar of URI is regular ({\it cf.} \cite{RANS}).
 Because the language of all URIs $U$ is regular, we can apply a
 ranking-based compression on a regular language.
 For example, the index of the URI $w_1 =$
 \url{http://dlt2014.sciencesconf.org/} is:
\[
 R_U(w_1) = 728552296504796066382113700758455910393907656035063493.
\]
The word $w_1$ is 32 bytes ($|w_1| = 32$), whereas its
 index $R_U(w_1)$ is 23 bytes ($\lfloor \log_{256} R_U(w_1) \rfloor = 23$).
 $w_1$ is compressed up to $72\%$ and, clearly, we can decompress
 it by the inverse of $R_U$ since ranking is bijective. \ed
\end{example}
In the case of a ranking on a regular language, the ranking function and
its inverse (\emph{unranking}) of $L$ can be calculated using the
adjacency matrix of the deterministic automaton $L$ ({\it cf.}
\cite{journals/mst/ChoffrutG95a,Berth:2010:CAN:1941063,RyomaSIN'YA2013}). Indeed,
Example \ref{ex:url} uses RANS\cite{RANS}, which is open source
software implemented by the author based on the algorithms in
\cite{RyomaSIN'YA2013}.

The computational complexity of an unranking function is higher than
a ranking function because the former requires matrix multiplication
but the latter does not ({\it cf.} Table 1 in \cite{RyomaSIN'YA2013}).
In Example \ref{ex:url}, the calculation of ranking
(compression) was performed in less than one second; however, the
calculation of unranking (decompression) took about two minutes.
The reason for such results is that the cost of matrix multiplication
is high (the naive algorithm has cubic complexity), and the unranking
algorithm requires matrix multiplications, while the ranking algorithm
does not.
The minimal automaton of $U$ used in Example \ref{ex:url} has
180 states, and its adjacency matrix multiplication cost is high in
practice. \\

\noindent{\bf Rank-one languages and our results}
There are several classes of matrices that have a matrix power that can be 
computed efficiently ({\it e.g.} diagonalisable matrices and low-rank
matrices).
We focus on \emph{rank-one matrices} from these classes. As we describe in
\sec{rankone}, the power of a rank-one matrix has constant time
complexity with linear-time preprocessing. 
We investigate \emph{rank-one languages}: the class of languages for
which the rank of minimal automaton is one. We define an automaton as
\emph{rank-$n$} if its adjacency matrix is rank-$n$.
Next, we introduce the definition of the rank of a language.
\begin{definition}\upshape\label{def:rank-n}
 A regular language $L$ is {\it rank-$n$} if there exists a rank-$n$
 deterministic automaton that recognises $L$, and there does not exist a
 rank-$m$ deterministic automaton that recognises $L$ for any $m$ less
 than $n$.
\ed
\end{definition}
However, Definition \ref{def:rank-n} raises the question of how to
find a minimal rank.
It is a classical theorem in automata theory that for any regular
language $L$, there is a unique automaton $\CA$ that recognises $L$ that
has a minimal number of states, and $\CA$ is called the \emph{minimal
automaton} of $L$. We intend to refine Definition \ref{def:rank-n} as
the following definition.

\setcounter{definition}{0}
\begin{definition}[refined]\upshape
 A regular language $L$ is {\it rank-$n$} if its minimal automaton is
 rank-$n$.
\ed
\end{definition}
Nevertheless, to achieve this we have to show that a minimal automaton has the
minimal rank for consistency of the above two definitions. Hence in
\sec{min}, we prove the following theorem, which has a more general
statement.

\begin{theorem}\label{theorem:min}
An automaton $\CA$ is minimal if and only if both the rank
 and the nullity of its adjacency matrix are minimal.
\ed
\end{theorem}
Theorem \ref{theorem:min} provides a necessary and sufficient condition for
the minimality of an automaton and is a purely algebraic characterisation of
minimal automata.
 This theorem is not obvious because, in general, for an automaton
 $\CA$, the number of states of $\CA$ and the rank (nullity) of $\CA$
 are not related. This is illustrated in Figure \ref{fig:threedfa},
 where the deterministic automaton $\CB_1$ has three states and its
 rank is two, whereas $\CC_1$ has four states and its rank is one, which
 equals the rank of the minimal automaton $\CA_1$. Therefore, we can-not
 argue naively that ``any minimal automaton has the minimal rank
 (nullity)'' by its minimality of states.

\begin{figure}[h!]
\begin{minipage}[t]{0.31\hsize}
\centering\includegraphics[width=1\columnwidth]{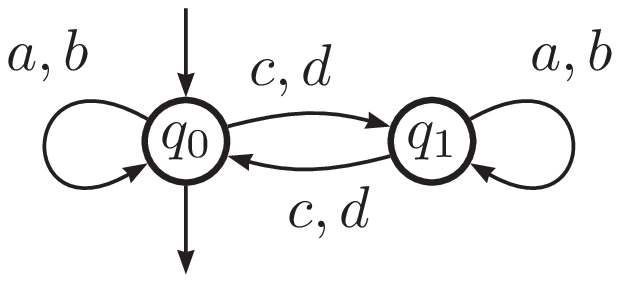}
{\large \centering{$\CA_1$}}
\end{minipage}
\begin{minipage}[t]{0.35\hsize}
\centering\includegraphics[width=.93\columnwidth]{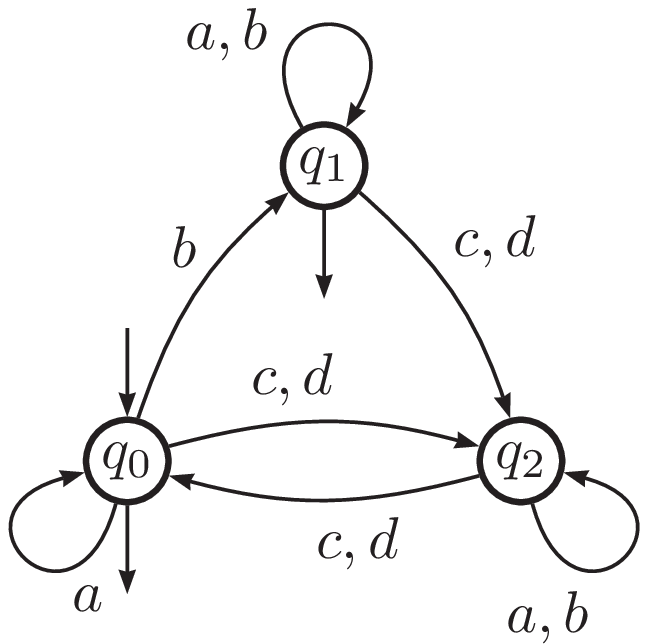}
{\large \centering{$\CB_1$}}
\end{minipage}
\begin{minipage}[t]{0.32\hsize}
\centering\includegraphics[width=1\columnwidth]{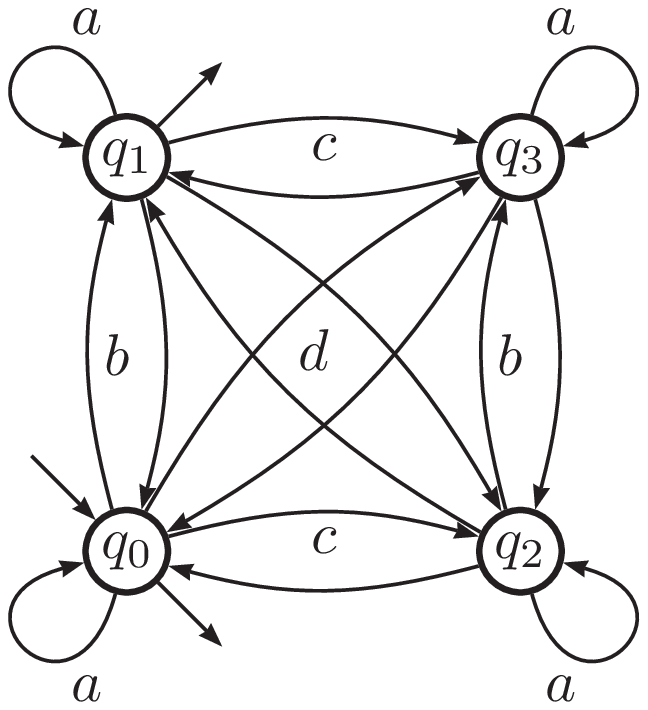}
{\large \centering{$\CC_1$}}
\end{minipage}
\[
 M(\CA_1) = \begin{bmatrix}
  2 & 2\\
  2 & 2
 \end{bmatrix}
 \;\;\;\;\;\;\;\;\;\;\;\;\;\;
 M(\CB_1) = \begin{bmatrix}
  1 & 1 & 2\\
  0 & 2 & 2\\
  2 & 0 & 2
 \end{bmatrix}
 \;\;\;\;\;\;\;\;\;\;\;\;\;\;\;\;\;\;\;
 M(\CC_1) = \begin{bmatrix}
  1 & 1 & 1 & 1\\
  1 & 1 & 1 & 1\\
  1 & 1 & 1 & 1\\
  1 & 1 & 1 & 1
 \end{bmatrix}
\]
 \caption{Three equivalent deterministic trim automata and these adjacency matrices.}
 \label{fig:threedfa}
\end{figure}

The proof consists of the use of two fundamental tools: \emph{equitable
 partition} from graph spectra theory and \emph{Nerode partition} from
 automata theory.
 We briefly introduce these two partitions in \sec{preliminaries}, then
 give the proof of Theorem \ref{theorem:min} in \sec{min}.
 In \sec{rankone}, we investigate the properties of rank-one languages
 and introduce expanded canonical automata.
 In \sec{discussion} we briefly discuss three topics that are not yet
 understood or lack maturity.
\section{Nerode partition and equitable partition}\label{sec:preliminaries}
We assume that the reader has a basic knowledge of automata, graphs and
linear algebra. 
All results in this section are well-known, and for more details, we
refer the reader to
\cite{Sakarovitch:2009:EAT:1629683} for automata theory and
\cite{Mieghem:2011:GSC:1983675} for graph spectra theory.

\subsection{Automata and languages}
A \emph{deterministic finite automaton} $\CA$ is a quintuple $\CA = \langle Q, A,
\delta, q_0, F \rangle$; the finite set of \emph{states} $Q$, the finite set $A$ called \emph{alphabet}, the \emph{transition function} $\delta: Q \times A \rightarrow
Q$, the \emph{initial state} $q_0$, and the set of \emph{final states}
$F \subseteq Q$. 
If $\delta(q, w) = p$ is a transition of automaton $\CA$, $w$ is said
to be the \emph{label} of the transition.
We call a transition $\delta(p, w) = q$ is \emph{successful} if its
destination is in the final states $F$ of $\CA$.
A word $w$ in $A^*$ is \emph{accepted} by $\CA$ if it is the label of a
successful transition from the initial state of $\CA$: $\delta(q_0, w) \in F$.
The symbol like $|\CA|$ denotes the number of states $|Q|$ for an automaton
$\CA$. The set of all acceptable words of $\CA$, or {\it language} of $\CA$,
is denoted by $L(\CA)$. We call two automata are {\it equivalent} if
their languages are identical.
An deterministic automaton $\CA = \langle Q, A, \delta, q_0, F \rangle$
is \emph{trim} if, for all state $q \in Q$, there exist two words $v$
and $w$ such that $\delta(q_0, v) = q$ (\emph{accessible}) and $\delta(q, w) \in
F$ (\emph{co-accessible}).

\subsection{Graphs and adjacency matrices}
A \emph{multidigraph} $G$ is a pair $G = \langle N, E \rangle$; the set
of \emph{nodes} $N$, the multiset of \emph{edges} $E$.
The \emph{adjacency matrix} $M(G)$ of $G$ is the $|N|$-dimensional
matrix defined as:
\[
 M(G)_{ij} := \text{the number of edges from node} \; i \; \text{to node} \; j.
\]
The \emph{spectrum} of a matrix $M$ is the multiset of the eigenvalues
of $M$ and is denoted by $\lambda(M)$.
The \emph{kernel} of a matrix $M$ is the subspace defined as $\{ \bm{v}
\mid M \bm{v} = \bm{0} \}$, and is denoted $\ker(M)$.
 We denote the dimension, rank and nullity (the dimension of the kernel)
 of $M$ by $\dim(M), \rank(M)$ and $\nullity(M)$ respectively.
The following dimension formula is known as the {\it rank-nullity theorem}:
\begin{eqnarray*}
\dim(M) = \rank(M) + \nullity(M).
\end{eqnarray*}

A \emph{partition} $\pi$ of a multidigraph $G = \langle N, E \rangle$ is a set of
nodal sets $\pi = \{ C_1, C_2, \ldots, C_k\}$ that satisfies the
following three conditions:
\begin{eqnarray*}
\emptyset \notin \pi \;\;\; \text{and} \;\;\;
 \displaystyle \bigcup_{C\in \pi} C = N \;\;\; \text{and} \;\;\;
\forall i, j \in [1, k], i \neq j \Rightarrow C_i \cap C_j = \emptyset.
\end{eqnarray*}
We call ${}_\pi M$ as the \emph{partitioned matrix} induced by $\pi$ of
$M$, that is partitioned as 
\begin{eqnarray}
 {}_\pi M = \begin{bmatrix}
			 M_{1,1} && \cdots && M_{1, k}\\
			 \vdots  &&        && \vdots\\
			 M_{k,1} && \cdots && M_{k, k}
			\end{bmatrix}\label{eq:partitionedmatrix}
\end{eqnarray}
where the block matrix $M_{i,j}$ is the submatrix of $M$ formed by the
rows in $C_i$ and the columns in $C_j$.
The \emph{characteristic matrix} $S$ of $\pi$ is the $|N| \times k$
matrix that is defined as follows:
\[
 S_{ij} = \begin{cases}
		   1 & \text{if} \; i \in C_j,\\
		   0 & \text{otherwise}.
		  \end{cases}
\]
In general, $S$ is a full rank matrix ($\rank(S) = \min(|N|, k) = k$) and $S^T S =
\diag(|C_1|, |C_2|, \ldots, |C_k|)$ where $S^T$ is the transpose of $S$.
The \emph{quotient matrix} $M^\pi$ of $M$ by $\pi$ is defined as the $k
\times k$ matrix:
\begin{eqnarray}
 M^\pi = (S^T S)^{-1} S^T M S \label{eq:quotient}
\end{eqnarray}
where $(S^T S)^{-1} = \diag(\frac{1}{|C_1|}, \frac{1}{|C_2|}, \ldots,
\frac{1}{|C_k|})$. That is, $(M^\pi)_{ij}$ denotes the average row sum
of the block matrix $({}_\pi M)_{i,j}$, in the intuitive sense of the
word.

\begin{example}\label{ex:quotient}
 Consider the deterministic automaton $\CB_1$ in Figure \ref{fig:threedfa}.
 Let $\pi_1$ be the partition of $\CB_1$: $\pi_1 = \{\{ q_0, q_1 \},
 \{q_2\}\}$, then its characteristic matrix $S$ and $S^T S$, the partitioned matrix
 ${}_{\pi_1} M(\CB_1)$ and the quotient matrix $M(\CB_1)^{\pi_1}$ are
 follows:
 \[
 S = \begin{bmatrix}
	  1 & 0 \\
	  1 & 0 \\
	  0 & 1 \\
	 \end{bmatrix},\;\;\;
 S^T S = \begin{bmatrix}
		  2 & 0\\
		  0 & 1
		 \end{bmatrix},\;\;\;
 {}_{\pi_1} M(\CB_1) =
 \begin{bmatrix}
  \begin{bmatrix}\vspace{-.2em}
   1 & 1\vspace{-.2em}\\
   0 & 2\vspace{-.2em}
  \end{bmatrix} &
  \begin{bmatrix}\vspace{-.2em}
   2 \vspace{-.2em}\\
   2\vspace{-.2em}
  \end{bmatrix}\\[.5em]
  \begin{bmatrix}
   2 & 0\vspace{-.2em}\\
  \end{bmatrix} &
  \begin{bmatrix}\vspace{-.2em}
   2
  \end{bmatrix}
 \end{bmatrix},
 \;\;\;
M(\CB_1)^{\pi_1} =
 \begin{bmatrix}
  2 & 2\\
  2 & 2
 \end{bmatrix}.
\]\ued \vspace{1em}
\end{example}

\subsection{Nerode partition}\label{sec:nerode}
Because an automaton $\CA$ can be regarded as a multidigraph, we can
naturally define the adjacency matrix, partitions and these quotient of
$\CA$ as the same manner.
 Let $q$ be a state of $\CA$. We denote by $F(q)$ the set of words $w$ that
 are labels of a successful transition starting from $q$. It is called the future of the
 state $q$. Two states $p$ and $q$ are said to be \emph{Nerode equivalent} if and only if
 $F(p) = F(q)$. \emph{Nerode partition} is the partition induced by Nerode equivalence.

Nerode's theorem states that states of minimal automaton are
blocks of Nerode partition, edges and terminal states are defined
accordingly ({\it cf.} \cite{nerode58,IGMA_BeaCro08}). That is, note
that the adjacency matrix of a minimal automaton equals the quotient
matrix of the adjacency matrix of an equivalent automaton by its Nerode
partition. For example, $\pi_1$ in Example
\ref{ex:quotient} is the Nerode partition of $\CB_1$ in Figure
\ref{fig:threedfa} and its induced quotient matrix $M(\CB_1)^{\pi_1}$ is
identical to the adjacency matrix  $M(\CA_1)$ of the minimal automaton
$\CA_1$ in the same figure.

\subsection{Equitable partition}\label{sec:equitalbe}
If the row sum of each block matrix $M_{i,j}$ in Equation
\eqref{eq:partitionedmatrix} induced by $\pi$ is constant, then the
partition $\pi$ is called \emph{equitable}.
 In that case the characteristic matrix $S$ of $\pi$ satisfies the
 following equation ({\it cf.} Article 15 in \cite{Mieghem:2011:GSC:1983675}):
\begin{eqnarray}
  M S = SM^\pi\label{eq:equitable}
\end{eqnarray}
If $\bm{v}$ is an eigenvector of $M^\pi$ belonging to the eigenvalue
$\lambda$, then $S \bm{v}$ is an eigenvector of $M$ belonging to the
same eigenvalue $\lambda$. Indeed, left-multiplication of the eigenvalue
equation $M^\pi \bm{v} = \lambda \bm{v}$ by $S$ yields:
\[
 \lambda S \bm{v} = (SM^\pi) \bm{v} = (MS) \bm{v} = M(S\bm{v}).
\]
For example, we can verify that $\pi_1$ in Example \ref{ex:quotient} is
equitable.
We conclude the the following lemma.
\begin{lemma}\label{lemma:equitable}
 Let $\pi$ be an equitable partition of a matrix $M$ and $M^\pi$ be its
 induced quotient matrix, then $\lambda(M^\pi) \subseteq \lambda(M)$
 holds. \qed
\end{lemma}

\begin{remark}
 Though many literature of graph spectra theory including \cite{Mieghem:2011:GSC:1983675}
 treat the notion of equitable partition on simple
 graphs, the properties of equitable partition including Equation
 \eqref{eq:equitable} and Lemma \ref{lemma:equitable} are also holds on
 multidigraphs without problems ({\it cf.} \cite{cvetkovic80}; Theorem 4.5).

 The concept of equitable partition was introduced in \cite{equitable}.
 Equitable partition have been considered in the literature also under
 the name \emph{divisor} and for more information the reader is referred
 to Chapter 4 in \cite{cvetkovic80}, where basic properties of divisor
 can be found. \ed
\end{remark}
\section{Minimal properties of minimal automata}\label{sec:min}
The ``if'' direction of the Theorem \ref{theorem:min} is obvious
from the rank-nullity theorem.
For proving the ``only if'' direction,
we prove the following two propositions.
\begin{enumerate}
 \item Quotient by an equitable partition always reduces the dimension,
	   rank and nullity, respectively (Proposition \ref{prop:equitable}).
 \item Nerode partition is equitable (Proposition \ref{prop:nerode}).
\end{enumerate}
Because, as we mentioned in \sec{nerode}, the adjacency matrix of
a minimal automaton equals to the quotient matrix of the adjacency
matrix of any equivalent automaton by its Nerode partition.

\begin{proposition}\label{prop:equitable}
 Let $\pi$ be an equitable partition of a matrix $M$ and $M^\pi$ be its
 induced quotient matrix, then the following inequalities hold.
\[
 \dim(M^\pi) \leq \dim(M), \;\;\; \rank(M^\pi) \leq \rank(M), \;\;\;
 \nullity(M^\pi) \leq \nullity(M). 
\]
\end{proposition}
\begin{proof}
$\dim(M^\pi) \leq \dim(M)$ is obvious, then we prove the rest two inequalities.

Let $\bm{v}$ be a vector in the kernel of $M^\pi$ and $\bm{w}$ be a
 vector not in the kernel of $M^\pi$, then the following equations hold.
\begin{eqnarray}
 MS \bm{v} = S(M^\pi \bm{v}) = \bm{0},\label{eq:kernel}\\
 MS \bm{w} = S(M^\pi \bm{w}) \neq \bm{0} \label{eq:nonkernel}. 
\end{eqnarray}
Equation \eqref{eq:nonkernel} is induced by $M^\pi \bm{w} \neq \bm{0}$
 and $S\bm{u} \neq \bm{0}$ for any $\bm{u} \neq \bm{0}$ since $S$ has
 full rank.
Equation \eqref{eq:kernel} and \eqref{eq:nonkernel} leads:
\[
 \bm{v} \in \ker(M^\pi) \Rightarrow S\bm{v} \in \ker(M) \;\;\;
 \text{and} \;\;\;
 \bm{w} \not \in \ker(M^\pi) \Rightarrow S\bm{w} \not \in \ker(M).
\]
For any linearly independent vectors $\bm{u}$ and $\bm{u}'$ then $S\bm{u}$
 and $S\bm{u}'$ are also linearly independent since $S$ has full rank.
This shows the rest two inequalities.
\qed
\end{proof}

\begin{proposition}\label{prop:nerode}
 Nerode partition is equitable.
\end{proposition}
\begin{proof}
 Let $\CA = \langle Q, A, \delta, q_0, F \rangle$ be a
 deterministic automaton and its Nerode partition $\pi = (C_1, C_2, \ldots, C_k); C_i
 \subseteq Q$.
 We prove by contradiction.

 Assume $\pi$ is not equitable, then there exist $C_i$ and $C_j$ in $\pi$
 and $p$ and $q$ in $C_i$ such that $p$ and $q$ have different number of
 transition rules into $C_j$. We assume without loss of generality that
 the number of transition rules into $C_j$ of p is larger than $q$'s.
 Then there exists at least one alphabet $a$ in $A$ such that $\delta(p, a) \in
 C_j$ and $\delta(q, a) \notin C_j$. 
 Let $p_a = \delta(p, a)$ and $q_a = \delta(q, a)$,
 then $p_a$ and $q_a$ are not Nerode equivalent since
 $p_a$ belongs to another partition of $q_a$'s.
 Hence $F(p_a) \neq F(q_a)$ holds and either $F(p_a)$ or
 $F(q_a)$ is not empty.
 We assume without loss of generality that $F(p_a)$ is not empty.
 Because $F(p_a) \neq F(q_a)$ and $F(p_a) \neq
 \emptyset$, there exists $w$ in $F(p_a)$ such that $w \notin F(q_a)$
 and then $w$ satisfies:
\begin{eqnarray*}
 \delta(p, a w) = \delta(p_a, w) \in F \;\;\;\; \text{and} \;\;\;\; \delta(q, a w) = \delta(q_a, w) \notin F.
\end{eqnarray*}
 This leads that $p$ and $q$ are not Nerode
 equivalent even though $p$ and $q$ belong to the same Nerode equivalent class
 $C_i$. This is contradiction.
 \qed
\end{proof}

It is also proved that any deterministic automaton includes the spectrum
of its equivalent minimal automaton by Proposition \ref{prop:nerode} and
Lemma \ref{lemma:equitable}.
\section{Rank-one languages and expanded canonical automata}\label{sec:rankone}
In this section, we focus on rank-one languages and introduce expanded
canonical automata.
Firstly, we introduce the well-known general properties of rank-one matrices ({\it
cf}. Proposition 1 in \cite{Osnaga2005}).

\begin{property}[characterization of a rank one matrix]\label{fact:rankone}
 Let $M$, $n \geq 2$, be a $n$-dimensional real matrix of rank one. Then
\begin{enumerate}
 \item There exists $\bm{x}, \bm{y}$ vectors in $\mathbb{C}^n; \bm{x},\bm{y} \neq
	   \bm{0}$ such that $M = \bm{x}\bm{y}^T$;
 \item $M$ has at most one non-zero eigenvalue with algebraic
	   multiplicity 1;
 \item This eigenvalue is $\bm{y}^T \bm{x}$.\ed
\end{enumerate}
\end{property}

Property \ref{fact:rankone} shows that, for any rank-one language $L$, its
counting function can be represented as a monomial: $C_L(n) = \alpha
\lambda^n$ for $n > 0$ and natural numbers $\alpha$ and $\lambda$.
In addition, rank-one matrices have beneficial property that their power can
be computable in constant time with linear-time preprocessing.
Indeed, for any $m$-dimensional rank-one matrix $M$, there exists $x,y$
such that $x y^T = M$ hence, the following equation holds for $\lambda =
y^T x$:
\[
 M^n = (\bm{x}\bm{y}^T)^n = \bm{x} (\bm{y}^T \bm{x})^{n-1} \bm{y}^T =
 \lambda^{n-1} \bm{x}\bm{y}^T = \lambda^{n-1}M.
\]
This shows that $(M^n)_{ij}$ equals $\lambda^{n-1} (\bm{x}\bm{y}^T)_{ij}$,
 and the inner product of $\bm{x}$ and $\bm{y}$ has linear-time complexity
 with respect to its dimension $m$.

\subsection{In-vector and out-vector}
For any rank-one matrix $M$, we can construct $\bm{x}$ and $\bm{y}$ such that
$M = \bm{x}\bm{y}^T$ from the ratio of the number of incoming edges and
outgoing edges, respectively.

\begin{definition}\upshape
 Let M be an $n$-dimensional rank-one matrix.
 The \emph{in-vector} of $M$ is a non-zero row vector having minimum
 length in $M$ and is denoted by $\inv(M)$.
 Because $M$ is rank-one, each row vector $\bm{v}_i$ in $M$ can be
 represented as $\bm{v}_i = \alpha_i \cdot \inv(M)$ for some natural
 number $\alpha_i \geq 1$. 
 The \emph{out-vector} of $M$ is the $n$-dimensional column vector that
 has an $i$-th element defined as the above coefficient $\alpha_i$ and
 denoted by $\out(M)$.

 By this construction, it is clear that an in-vector and out-vector
 satisfy $\out(M) \cdot \inv(M) = M$. \ued
\end{definition}

In general, an automaton $\CA$ may have the state $q$ such that there are
no transition rules into $q$.
Hence, the in-vector must be taken from non-zero vectors in the given
matrix. We note that, for any rank-one matrix $M$, the out-vector of $M$
always contains one because it consists of the coefficients of the
in-vector of $M$ ({\it cf.} Example \ref{ex:zero}).

\begin{example}\label{ex:zero}
Consider the rank-one automaton shown in the adjacent figure.

\hspace{-1.5em}\begin{minipage}{0.7\hsize}\vspace{-1.7em}
This automaton is deterministic and trim.
Its adjacency matrix $M$ and in-vector $\inv(M)$ and out-vector $\out(M)$ are
 follows.
\[
 M = \begin{bmatrix}
  0 & 2 & 1\\
  0 & 4 & 2\\
  0 & 2 & 1
 \end{bmatrix}, \;\;\;
 \inv(M) = \begin{bmatrix}
  0, 2, 1
 \end{bmatrix},\;\;\;
 \out(M) =
 \begin{bmatrix}
  1\\
  2\\
  1
 \end{bmatrix}.
\]
\end{minipage}
\begin{minipage}{0.3\hsize}
\includegraphics[width=.85\columnwidth]{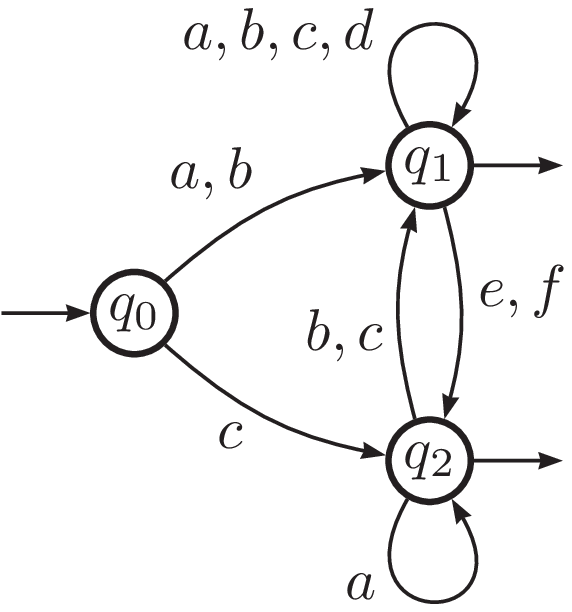}
\end{minipage}
Note that $\inv(M)_0 = 0$ means that $q_0$ has no incoming transition rules. 
\ed
\end{example}

\subsection{Expanded canonical automata}
First, we define a normal form of a rank-one automaton.

\begin{definition}\upshape\label{def:normlform}
A rank-one automaton $\CA$ is \emph{expanded normal} if for its
 adjacency matrix $M$, each element of the in-vector of $M$ equals to
 zero or one.
\ed
\end{definition}
The automaton in Figure \ref{ex:zero} is not expanded normal because the
second element of its in-vector equals two.
Expanded normal form is a graph normal form of automata, and does not
consider labels. 

Secondly, we propose the operation {\it expansion} that expands the
 given matrix (graph or automaton) algebraically.

\begin{definition}\upshape\label{def:expansion}
 Let ${}^\pi M$ and $M$ be two matrices of dimension $m$ and $n$,
 respectively.
 We define ${}^\pi M$ as an expansion of M if there exists a partition
 $\pi = \{ C_1, C_2, \ldots, C_n \}$ of ${}^\pi M$ such that the
 characteristic matrix $S$ of $\pi$ satisfies:
\[
 {}^\pi M = S M (S^T S)^{-1} S^T.
\]
\ued
\end{definition}
Expansion is an algebraic transformation that increases the dimension
of the given matrix. Intuitively, expansion can be regarded as an inverse
operation of quotient.
Indeed, for any expanded matrix ${}\pi M$ of some $n$-dimensional $M$ by
$\pi$ and its characteristic matrix $S$, we have the following equation:
\begin{eqnarray*}
 ({}^\pi M)^\pi = (S^T S)^{-1} S^T ({}^\pi M) S = (S^T S)^{-1} S^T \left(S M
  (S^T S)^{-1} S^T \right)S = M
\end{eqnarray*}
which holds by Equation \eqref{eq:quotient} and Definition
\ref{def:expansion}.
If $M$ is rank-one, then for any expanded matrix ${}^\pi M$ of
$M$, the out-vector of ${}^\pi M$ consists of same elements as those of
the out-vector of $M$. This reflects the invariance of the number of
outgoing transition rules of the Nerode equivalent states ({\it cf.}
Figure \ref{fig:twograph}).

\begin{figure}[h!]
\begin{minipage}[t]{0.3\hsize}
\centering\includegraphics[width=1\columnwidth]{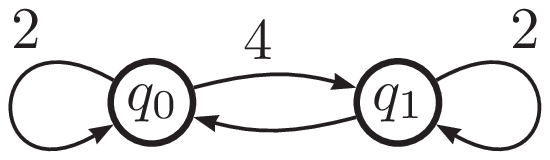}\\
{\large \centering{$\CD_1$}}
\end{minipage}
\begin{minipage}[t]{0.25\hsize}
\centering{ $\xrightarrow{\text{Expand by} \; \{\{q_0\}, \{q_1, q'_1\}\} \;\;
 }$}
\centering{ $\xleftarrow[\text{Quotient by} \; \{\{q_0\}, \{q_1, q'_1\}\}]{} $}
\end{minipage}
\begin{minipage}[t]{0.4\hsize}
\centering\includegraphics[width=.83\columnwidth]{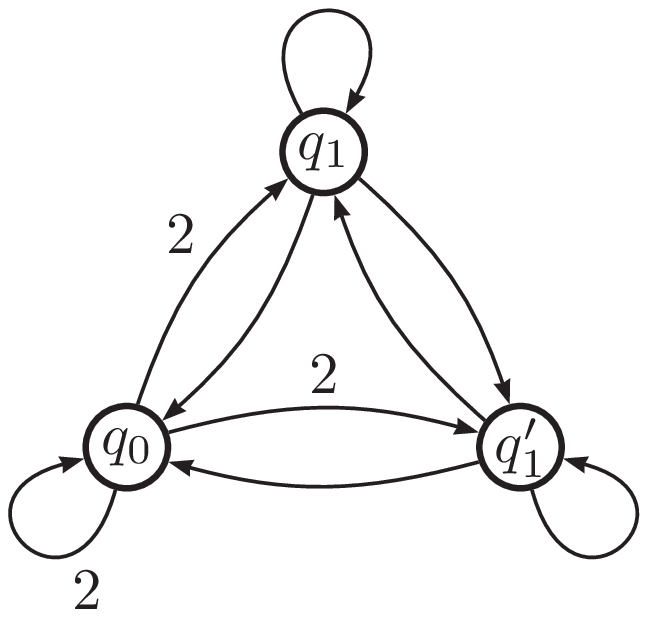}\\
{\large \centering{${}^\pi\CD_1$}}
\end{minipage}
\[
 \begin{bmatrix}
  2 & 4\\
  1 & 2
 \end{bmatrix} =
 \begin{bmatrix}
  2\\
  1
 \end{bmatrix}
 \begin{bmatrix}
  1 & 2
 \end{bmatrix}
 \;\;\;\;\;\;\;\;\;\;\;\;\;\;\;\;\;\;\;\;\;\;\;\;\;\;\;\;
 \;\;\;\;\;\;\;\;\;\;\;\;\;\;\;\;\;\;\;\;\;\;\;\;\;\;\;\;
 \begin{bmatrix}
  2 & 2 & 2\\
  1 & 1 & 1\\
  1 & 1 & 1
 \end{bmatrix} =
 \begin{bmatrix}
  2\\
  1\\
  1
 \end{bmatrix}
 \begin{bmatrix}
  1 & 1 & 1
 \end{bmatrix}
\]
 \caption{The rank-one graph $\CD_1$ and its expanded
 canonical form ${}^\pi\CD_1$.}
 \label{fig:twograph}
\end{figure}

Finally, we define a canonical automaton of a rank-one language:
\emph{expanded canonical automaton}.
The minimal automaton of a regular language $K$ is uniquely determined
by $K$, whereas the expanded canonical automaton of a
rank-one language $L$ is not uniquely determined, but its graph
structure is uniquely determined by $L$.

\begin{definition}\label{def:canonical}\upshape
Let $L$ be a rank-one language, then we define its expanded canonical
 automaton ${}^\pi \CA_L$ as the expanded automaton of the minimal automaton
 $\CA_L$ of $L$ by a partition $\pi = \{ C_1, C_2, \ldots C_{|\CA_L|}\}$
 such that, for all $C_i \in \pi$,  $|C_i| = \inv(\CA_L)_i$ if
 $\inv(\CA_L)_i \neq 0$ and 1 otherwise.
\ed
\end{definition}
By the definition, it is clear that for any rank-one language $L$,
its expanded canonical automaton is expanded normal ({\it cf.} Figure
\ref{fig:twograph}, or $\CA_1$ and its expanded canonical automaton
$\CC_1$ in Example \ref{fig:threedfa}).
As we describe in \sec{closure}, we introduce expanded canonical
automata  for analysis and evaluation of the closure properties of rank-one
languages.
Because of the limitations of space, a detailed discussion of expanded canonical
automata is not possible here.
\section{The way for further developments}\label{sec:discussion}
\subsection{Closure property of rank-one languages and
  decomposability}\label{sec:closure}
It is natural to consider the closure properties of rank-one languages.
However,  with some exceptions ({\it e.g.}
\emph{quotient}, \emph{prefix}\footnote{The operations that can be
realised without destroying graph structure of a deterministic
automaton}), the class of lank-one languages is, for the most part, not
closed under an operation on languages: {\it e.g.}
\emph{union}, \emph{concatenation} and \emph{Kleene star}.
Indeed, for the two rank-one (expanded canonical) automaton $\CE_1$ and
$\CF_1$ in Figure \ref{fig:twodfa}, the union of $L(\CE_1)$ and $L(\CF_1)$
has the spectrum $\{3, 2, -1, -1, -1\}$ (without zeros) and is
rank-five. We note that $L(\CE_1)$ and $L(\CF_1)$ have the same prefix, and
the minimal automaton of $L(\CE_1) \cup L(\CF_1)$ is strongly connected
and has nine states. In addition, note that there exist rank-one
languages $K and L$ such that the union language $K \cup L$ has
irrational and complex eigenvalues.

\begin{figure}[h!]
\begin{minipage}[t]{0.5\hsize}
\centering\includegraphics[width=.65\columnwidth]{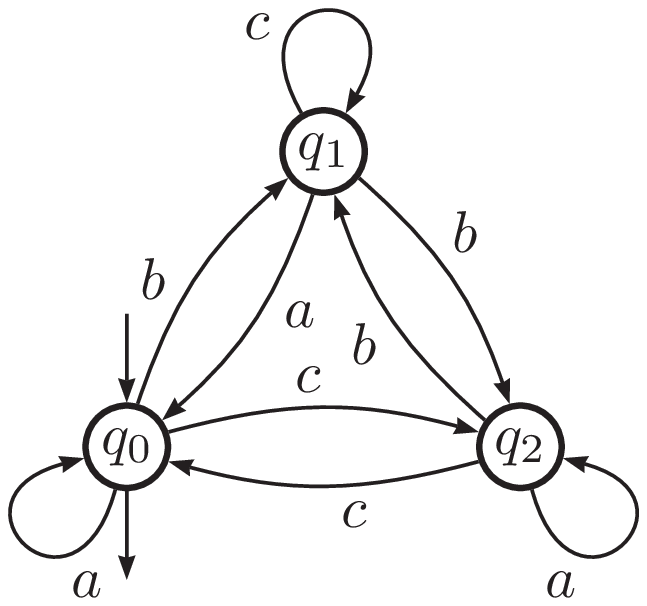}
{\large \centering{$\CE_1$}}
\end{minipage}
\begin{minipage}[t]{0.5\hsize}
\centering\includegraphics[width=.65\columnwidth]{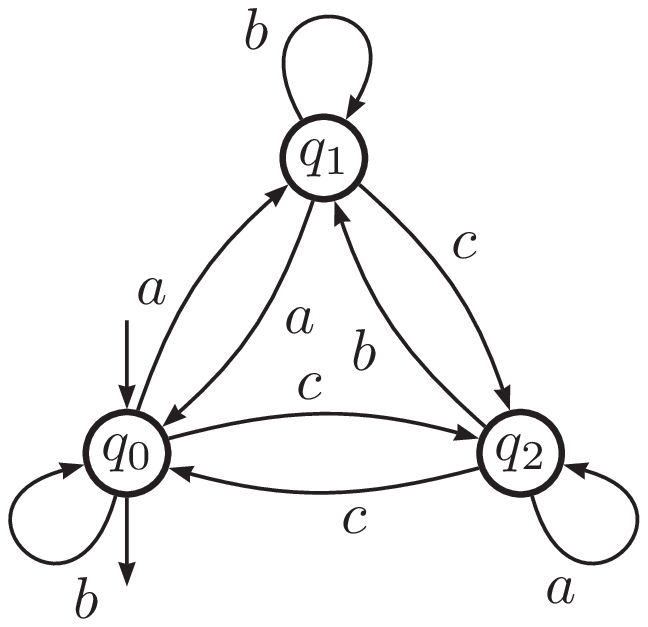}
{\large \centering{$\CF_1$}}
\end{minipage}
 \caption{Two rank-one automata (expanded canonical automata).}
 \label{fig:twodfa}
\end{figure}

Conversely, we consider the closure of rank-one languages with an
operation on languages or decomposability into rank-one languages
(\emph{rank-one decomposition}).
In the case of matrices, matrix rank-one decomposition is well studied
and there exist fundamental results such as \emph{orthogonal
decomposition} for real symmetric matrices.
We are interested in investigating regular language rank-one
decomposition.
 
\subsection{Rank of unambiguous automata}
The class of \emph{unambiguous automata} is a more general class of automata than the
class of deterministic automata ({\it cf.} \cite{Sakarovitch:2009:EAT:1629683}).
We intend to generalise Definition \ref{def:rank-n} as Definition
\ref{def:urank-n} which is more essential for the counting structure of
languages because unambiguous automata is the most general class that
satisfies Equality \eqref{eq:counting}.
It will be interesting to determine whether the rank of a minimal unambiguous
automaton is minimal. If so, we can refine Definition \ref{def:urank-n}
in a similar manner.

\begin{definition}\upshape\label{def:urank-n}
 A regular language $L$ is {\it unambiguous rank-$n$} if there exists a
 rank-$n$ unambiguous automaton recognises $L$ and does not exist a
 rank-$m$ unambiguous automaton recognises $L$ for any $m$ less than
 $n$.
\ed
\end{definition}

\subsection{Relation between the conjugacy of automata}
B\'eal {\it et al.} developed the theory of \emph{conjugacy of automata}
({\it cf.}
\cite{Beal:2005:EZA:2104063.2104105,Beal:2006:CEW:2094874.2094885}) that
gives structural information on two equivalent $\mathbb{K}$-automata.
Conjugacy of automata is a theory based on matrices, and we think some
results in this paper may be reconstructed by the theory of conjugacy.\\

\newpage
\noindent{\bf Acknowledgement}
I would like to thank my adviser, Kazuyuki Shudo, for his continuous
support and encouragement. Special thanks also go to Yuya Uezato, is 
a postgraduate student at University of Tsukuba, who provided carefully
considered feedback and valuable comments.

\bibliographystyle{splncs.bst}
\bibliography{ref}
\end{document}